\documentclass[a4paper, 10pt]{cas-sc}
\usepackage[authoryear]{natbib}
\def\tsc#1{\csdef{#1}{\textsc{\lowercase{#1}}\xspace}}
\tsc{WGM}
\tsc{QE}
\tsc{EP}
\tsc{PMS}
\tsc{BEC}
\tsc{DE}

\ExplSyntaxOn
\keys_set:nn { stm / mktitle } { nologo }
\ExplSyntaxOff

\newcommand{\significancestatement}[1]{}  %
\newcommand{\correspondingauthor}[1]{}
\newcommand{\authorcontributions}[1]{}
\newcommand{\authordeclaration}[1]{}
\newcommand{\equalauthors}[1]{}
\newcommand\myequationstyle{\displaystyle}
\newcommand\myoptionalpart[2]{#2}

\usepackage[australian]{babel}

\usepackage[utf8]{inputenc}
\usepackage{array}
\usepackage{mathtools}
\usepackage{amssymb}
\usepackage{amsmath}
\usepackage{amsfonts}
\usepackage{amsthm}
\usepackage{subfigure}
\usepackage{enumitem}
\usepackage{afterpage}
\usepackage{placeins}

\theoremstyle{plain}
\newtheorem{theorem}{Theorem}
\newtheorem{proposition}{Proposition}

\newtheorem{remark}{Remark}
\newtheorem{example}{Example}
\theoremstyle{definition}

\newcommand{\Gini}{\mathcal{G}}
\newcommand{\Bonf}{\mathcal{B}}
\newcommand{\Hoov}{\mathcal{H}}
\newcommand{\Verg}{\mathcal{V}}
\newcommand{\Pequ}{\mathcal{P}}

\begin{document}

\author[a]{Lucio Bertoli-Barsotti}[orcid=0000-0003-4508-3172]
\ead{lucio.bertoli-barsotti@unibg.it}
\credit{Conceptualisation of this study, Methodology,
Writing}

\author[b,c]{Marek Gagolewski}[orcid=0000-0003-0637-6028]
\ead{marek.gagolewski@pw.edu.pl}
\ead[url]{https://www.gagolewski.com}
\cormark[1]
\cortext[cor1]{Corresponding author}
\credit{Conceptualisation of this study, Methodology,
Data Curation, Investigation, Software,
Writing}

\author[d]{Grzegorz Siudem}[orcid=0000-0002-9391-6477]
\ead{grzegorz.siudem@pw.edu.pl}
\ead[URL]{http://if.pw.edu.pl/~siudem}
\credit{Conceptualisation of this study, Methodology,
Writing}

\author[c]{Barbara Żogała-Siudem}[orcid=0000-0002-2869-7300]
\ead{zogala@ibspan.waw.pl}
\credit{Data Curation, Investigation, Visualisation, Software, Writing}

\shortauthors{Bertoli-Barsotti et al.}

\address[a]{University of Bergamo, Department of Economics, Italy}

\address[b]{Warsaw~University of Technology, Faculty of Mathematics and Information Science, ul. Koszykowa 75, 00-662 Warsaw, Poland}

\address[c]{Systems Research Institute, Polish Academy of Sciences,
ul. Newelska 6, 01-447 Warsaw, Poland}

\address[d]{Warsaw~University of Technology, Faculty of Physics,
ul. Koszykowa 75, 00-662 Warsaw, Poland}

\let\WriteBookmarks\relax
\def\floatpagepagefraction{1}
\def\textpagefraction{.001}

\shorttitle{Equivalence of inequality indices in the three-dimensional model~of~informetric impact}
\title[mode = title]{Equivalence of inequality indices in the three-dimensional model~of~informetric impact}

\begin{abstract}
Inequality is an inherent part of our lives: we see it in the distribution of incomes, talents, citations, to name a few. However, its intensity varies across environments: there are systems where the available resources are relatively evenly distributed but also where a small group of items or agents controls the majority of assets. Numerous indices for quantifying the degree of inequality have been proposed but in general, they work quite differently.

We recently observed (Siudem et al., PNAS 117:13896--13900, 2020) that many rank-size distributions might be approximated by a time-dependent agent-based model involving a mixture of preferential (rich-get-richer) and accidental (sheer chance) attachment. In this paper, we point out its relationship to an iterative process that generates rank distributions of any length and a predefined level of inequality, as measured by the Gini index.

We prove that, under our model, the Gini, Bonferroni, De Vergottini, and Hoover indices are equivalent for samples of similar sizes. Given one of them, we can recreate the value of another measure. Thanks to the obtained formulae, we can also understand how they depend on the sample size. An empirical analysis of a large database of citation records in economics (RePEc) yields a good match with our theoretical derivations.
\end{abstract}


\maketitle

\noindent
Please cite this paper as:
Bertoli-Barsotti, L., Gagolewski, M., Siudem, G., Żogała-Siudem, B., Equivalence of inequality indices in the three-dimensional model of informetric impact, \textit{Journal of Informetrics} \textbf{18}(4), 101566, 2024, \href{https://doi.org/10.1016/j.joi.2024.101566}{DOI:10.1016/j.joi.2024.101566}.

\section{Introduction}

Given a series of measurements, indicators, scores, counts, or any other numeric values, it is natural to order them from the highest to the lowest. This way, we get better insight into the aspects of reality they aim to capture. In particular, various rankings (e.g., of universities, movies, or restaurants; see \citealp{iniguez2022dynamics}) promise to make our lives easier by claiming they can separate seeds from the chaff. Investigation and prediction of the size of the top or otherwise extreme values is crucial in risk analysis or disaster prevention \citep{voitalov2019scale,pickands1975statistical,marshall2007life}. The search for patterns and universalities in sorted data remains a fundamental, multidisciplinary research topic \citep{holme2022universality,Newman2005}. This includes the study of ranking dynamics \citep{iniguez2022dynamics} and the distribution of their static snapshots, from the most straightforward Zipf power-law models \citep{Newman2005} to more complex ones \citep{Petersen2011,PNAS2020,pricepareto,singh2022quantifying}. Rankings are inextricably linked to inequality. Unfortunately, not everyone can stand on the podium; often, the winner takes it all \citep{Perc}, especially in the context of skewed distributions commonly observed in informetrics \citep{deSollaPrice1965,Newman2005}. To quote \citet[p.~308]{rousseau2018becoming} ``the informetric laws describe situations in which a large inequality is present''. For example, long--tailed distributions, such as those described by the power laws \citep{Newman2005} or Pareto-type distributions \citep{arnold2015:pareto,pricepareto}, are typically used in informetrics to model information production processes (IPPs; systems consisting of ``sources'' producing ``items''; \citealp{egghe1990introduction,egghe2005power}) where we observe a high concentration of assets amongst only a few top scorers. A concentration measure is one of the two elements, the other being {\it the production} of the theory on impact \citep{egghe2022rank,egghe2023global}, and the Lorenz order is the most direct way to represent the concentration relationship – which is ultimately a method for establishing rankings among IPPs.

The recently-proposed 3DI model (three dimensions of impact; \citealp{PNAS2020,pricepareto}), can be conceived of as a rank-size approach to describing the mechanisms governing the growth of bibliographic and other networks studied originally by \citet{deSollaPrice1965}.
Namely, consider a process where, in every time step, a system (e.g.,~a citation network, a cluster of internet portals) grows by one entity (e.g., a new paper, a website). In each iteration, we distribute $m$ impact or wealth units (e.g., citations, links) amongst the already-existing entities:
\begin{itemize}[nosep]
\item
$a=(1-\rho)m$ units totally at random,
\item
$p=\rho m$ units according to the preferential attachment rule,
\end{itemize}
with $\rho$ representing
the extent to which the rich-get-richer rule dominates over pure luck.  Inspired by the observation in \citep{bertoli2022equivalent,OurGiniLorenz}, in Section~\ref{sec:giniprocess}, we will show that the $\rho$ parameter naturally corresponds to the value of the Gini index of the resulting ordered sample of impact measures. We will thus indicate the relationship between the degree of preferentiality and inequality.

Different systems or environments naturally have varied sizes and levels of inequality \citep{cowell2000measurement,silber2012handbook}, i.e., what percentage of the top scorers is in possession of the majority of the resources. As the degree of evenness vs monopoly can be measured by different indicators (e.g., the Gini, Bonferroni, or Hover index), a question of which one is the most informative often arises \citep{prathap2014zynergy,prathap2022comments}.
In Section~\ref{sec:inequality}, we will show the equivalence of many popular inequality indices in our model by expressing them in terms of monotone, 1-to-1 functions of the Gini index.

In Section~\ref{sec:experiments}, an analysis of a large number of citation records consisting of research papers in economics (RePEc) will show a quite good match between the empirical data and our theoretical derivations.

\section{The Gini-stable process, rich-get-richer, and random distribution of wealth}\label{sec:giniprocess}

In this section, after introducing some basic notation convention, we recall the iterative affine process proposed by \citet{bertoli2022equivalent}, which \citet{OurGiniLorenz} proved to preserve the Gini index throughout all its iterations. Moreover, it is the only such process, as shown in Section \ref{subsec:uniqueness}. In Section \ref{subsec:3DSI}, we prove that it is equivalent to the 3DI model by \citet{PNAS2020} but in a different parametrisation, given by the Gini index. In Section \ref{subsec:lorenz}, we
connect the normalised vectors $\boldsymbol{p}^{(N,G)}$ with the Lorenz curves.

\subsection{Notation}\label{subsec:notation}

We utilise the following notation.
    The gamma function is given by $\Gamma(z)=\int_0^\infty t^{z-1} e^{-t}\,dt$, $\Gamma(z+1)=z\Gamma(z)$,
    the polygamma functions are defined by
$\psi^{(m)}(z)=\frac{d^{m+1}}{dz^{m+1}}\log \left(\Gamma(z)\right)$,
    the digamma function is $\psi(z)=\psi^{(0)}(z)=\Gamma'(z)/\Gamma(z)$,
    the harmonic number $H_n=\sum_{k=1}^n k^{-1}=\psi(n+1)-\gamma$, where $\gamma \approx 0.577$ is the Euler constant.

\subsection{The derivation of the Gini-stable process}\label{subsec:derivation}

For any fixed parameter $G\in[0,1)$, consider an iterative process discussed
in \citep{OurGiniLorenz} whose update formula features a simple affine function
of the consecutive elements:

\begin{equation}\label{eq:process}
p_k^{(N,G)} = a_N(G) + b_N(G)\, p_k^{(N-1,G)}, \qquad k=1,\dots,N,
\end{equation}

\noindent
for $N\geq 2$, under the assumption that $p_N^{(N-1,G)} = 0$ and $p_1^{(1,G)}=1$
and:

\begin{eqnarray}
a_N(G) &=& \frac{1-G}{G(N-2)+1} \frac{1}{N}, \label{eq:at}\\
b_N(G) &=& \frac{G(N-1)}{G(N-2)+1}. \label{eq:bt}
\end{eqnarray}

\noindent
Then, for any $N\ge 2$,
$\boldsymbol{p}^{(N,G)}=\left(p_1^{(N,G)},\dots,p_N^{(N,G)}\right)$
is an \emph{ordered normalised positive $N$-vector}, i.e.,
$\sum_{k=1}^N p_k^{(N,G)}=1$
and
$1\ge p_1^{(N,G)}\ge p_2^{(N,G)}\ge\dots\ge p_N^{(N,G)}\ge 0$.

\begin{figure}[t!]

\centering
\includegraphics[width=\myoptionalpart{0.95}{0.7}\linewidth]{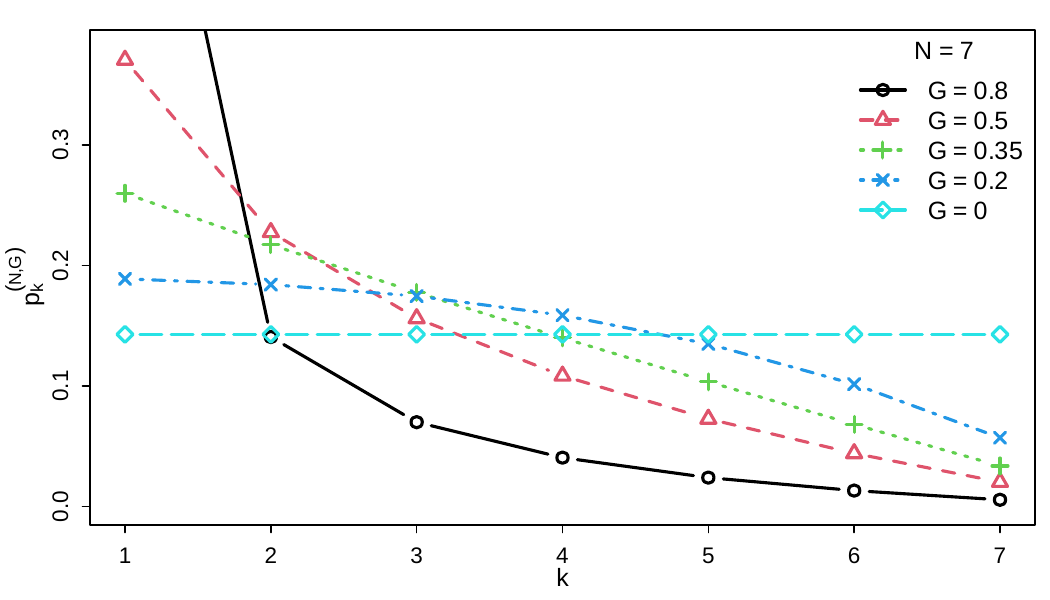}

\caption{\label{fig:model-vectors} Example ordered normalised $N$-vectors
generated using the affine update formula given by Eq.~(\ref{eq:process}).
Parameter $G$ controls the level of inequality, as measured
by the Gini index.}
\end{figure}

\begin{example}
Here are two example outcomes for $G=0.25$ and $G=0.75$
up to $N=4$:

\begin{center}
\begin{tabular}{lll}
      & $\boldsymbol{p}^{(N,0.25)}$             & $\boldsymbol{p}^{(N,0.75)}$  \\
$N=2$ & $( 0.625, 0.375 )$                      & $( 0.875, 0.125 )$  \\
$N=3$ & $( 0.450, 0.350, 0.200 )$               & $( 0.798, 0.155, 0.048 )$  \\
$N=4$ & $( 0.350, 0.300, 0.225, 0.125 )$        & $( 0.743, 0.164, 0.068, 0.025 )$  \\
\end{tabular}
\end{center}
\end{example}

There exists an explicit formula for the individual components
of the normalised vectors. Namely, we can show that (see
Sec.~\ref{subsec:3DSI} and \citealp{OurGiniLorenz}):

\begin{equation}\label{eq:pk}
p_k^{(N,G)}=
\begin{cases}
\myequationstyle
\frac{1}{N} \frac{1-G}{2G-1} \left(
\frac{\Gamma(N+1)}{\Gamma(N+1+1/G-2)} \frac{\Gamma(k+1/G-2)}{\Gamma(k)}
-1
\right),
&
\myoptionalpart{}{\text{if }} G\neq \frac{1}{2},
\\
\myequationstyle
\frac{1}{N}  \left(
H_N-H_{k-1}
\right)=\frac{1}{N}\sum_{i=k}^N\frac{1}{i},
&
\myoptionalpart{}{\text{if }} G = \frac{1}{2}.
\\
\end{cases}
\end{equation}

\subsection{Gini index is exactly $G$.}\label{subsec:gini}

Figure~\ref{fig:model-vectors} depicts a few example normalised vectors
$\boldsymbol{p}^{(7,G)}$ for different $G$s.
We see that the $G$ parameter controls the level of inequality
of the data distribution: $G=0$ yields all elements equal to $1/N$,
and, as $G\to 1$, all mass is transferred to the first element.

It turns out that our process generates ordered normalised
vectors whose Gini's index, $\Gini\left(\boldsymbol{p}^{(N,G)}\right)$ (\citealp{Gini1912:index}; see also, e.g., \citealp*{ginimeth}),
is exactly equal to $G$. Namely, for
$N\ge 2$ and $G\neq \frac{1}{2}$
we have:

\begin{eqnarray*}
\Gini\left(p_1^{(N,G)},\dots,p_N^{(N,G)}\right)
&=&
  \frac{1}{N-1} \sum_{k=1}^N (N-2k+1) p_{k}^{(N,G)}
\\
&=&
\myequationstyle
\frac{1-G}{(N-1)(2G-1)}
\frac{G (2 G-1) (N-1)}{1-G}
= G,
\end{eqnarray*}

\noindent
and  for $G=\frac{1}{2}$ it holds:

\begin{eqnarray*}
\Gini\left(p_1^{(N,1/2)},\dots,p_N^{(N,1/2)}\right)
= \frac{(N-1)N}{2(N-1)N}=\frac{1}{2}.
\end{eqnarray*}

\noindent
Thus, quite remarkably, $N$ (size) and $G$ (inequality)
are two \textit{independent} parameters in our model.
This is why we refer to the above as the \emph{Gini-stable process}.

\subsection{It is the only affine process of this kind.}\label{subsec:uniqueness}
Our process transforms an ordered normalised vector
to a new, longer, ordered normalised vector of the same Gini index.
What is more, the aforementioned coefficients $a_N$ and $b_N$ make up
the only affine transformation that achieves this property.

\begin{proposition}
For any $N\ge 3$ and  $G\in[0,1)$, let
$\boldsymbol{q}^{(N-1)}$
be an ordered probability $(N-1)$-vector
with the Gini index of $\Gini\left(q_1^{(N-1)},\dots,q_{N-1}^{(N-1)}\right) = G$.
Then, generating the normalised $N$-vector
$\boldsymbol{q}^{(N)}$ using the update formula
given by Eq.~(\ref{eq:process})
yields $\Gini\left(q_1^{(N)},\dots,q_N^{(N)}\right) = G$
if and only if $a_N$ and $b_N$ are given exactly by
Eq.~(\ref{eq:at}) and Eq.~(\ref{eq:bt}), respectively.
\end{proposition}

\begin{proof}
We have
$\Gini\left(q_1^{(N-1)},\dots,q_{N-1}^{(N-1)}\right)=
\frac{1}{N-2} \sum_{k=1}^{N-1} (N-2k) q_{k}^{(N-1)}
=G
$.
Moreover,
\begin{eqnarray*}
\Gini\left(q_1^{(N)},\dots,q_{N}^{(N)}\right)
&=&
\frac{1}{N-1} \sum_{k=1}^N (N-2k+1) q_{k}^{(N)}
\\
&=&
\frac{1}{N-1} \sum_{k=1}^N (N-2k+1) (a_N+b_N q_{k}^{(N-1)})
\\
&=&
\myoptionalpart{}{
\frac{a_N}{N-1} \sum_{k=1}^N (N-2k+1)
+
\frac{b_N}{N-1} \sum_{k=1}^{N-1} (N-2k+1)  q_{k}^{(N-1)}
\\
&=&
}
0
+
\frac{b_N}{N-1} \sum_{k=1}^{N-1} (N-2k)  q_{k}^{(N-1)} + \frac{b_N}{N-1}
\\
&=&
b_N \left( \frac{N-2}{N-1} G + \frac{1}{N-1} \right)
\end{eqnarray*}
is equal to $G$
if and only if
$b_N = \frac{G (N-1)}{G (N-2)+ 1 }$,
which is exactly Eq.~(\ref{eq:bt}).

We also need to have $\sum_{k=1}^N p_{k}^{(N)} = 1$.
But after some simple transformations,
one can show that this holds if and only if $a_N = (1 - b_N)/N$,
which corresponds to Eq.~(\ref{eq:at}).
\end{proof}

\subsection{Relation to the 3DI model.}\label{subsec:3DSI}
To strengthen the model's underpinnings,
let us return to the 3DI (three dimensions of impact) model \citep{PNAS2020}
mentioned in the introduction.
Let $X_k(t)$ denote the impact of the $k$-th richest entity at time step $t$
(e.g., the number of citations to the $k$-th most cited paper).
We assume ${X}_k(k-1)=0$ for every $k$, i.e., the $k$-th object enters the
system with no impact units.
The update formula in our model is a mixture of the accidental and rich-get-richer components governed by the $\rho$ parameter\footnote{\citet{PNAS2020,pricepareto} only studied the case of $\rho\in(0,1)$.
However, let us note that $\rho<0$ is not only possible but also has a nice interpretation
\citep{gagolewski2022ockham,bertoli2022equivalent}. In such a scenario,
we initially distribute more than the assumed $m$ citations at random,
but then we take away from those who are already rich (rich get less).}:

\begin{equation}
{X}_k{(t)} =
\underbrace{{X}_k{(t-1)}}_{\mathrm{previous\; value}}+
\underbrace{\frac{a}{t}}_{\mathrm{accidental\;income}}+
\underbrace{
    p\,\frac{{X}_k{(t-1)}+\frac{a}{t}}{ (t-1)m + a}}_{\mathrm{preferential\;gain\;or\;loss}
},\label{eq:master}
\end{equation}

\noindent
where $a=(1-\rho)m$ and $p=\rho m$. Let us note that there is some similarity between the preference/randomness structure in the 3DI model and the SJR indicator \citep{gonzalez2010new} which is a modification of the PageRank algorithm by \citet{brinpage}.

If we assume $\rho=0$, then we are only left with the accidental component,
and our model reduces to the harmonic one (compare \citealp{cena2022validating}):

\begin{equation*}
   {X}_k(t) = m \sum\limits_{i=k}^t\frac{1}{i} = m \left(H_t-H_{k-1}\right).
\end{equation*}

\noindent
Note that ``purely accidental'' does not mean that every entity ends up with the same amount of wealth, as older agents have had more opportunities to become impactful (``the old get richer'').
On the other hand, for $\rho<1$ and $\rho\neq 0$, the solution is:

\begin{eqnarray}
X_k(t)
&=&
m\frac{1-\rho}{\rho}
\myequationstyle
\left(
\frac{\Gamma(t+1)}{\Gamma(t+1-\rho)}\frac{\Gamma(k-\rho)}{\Gamma(k)} - 1
\right)
\nonumber
\\
&=&
m\frac{1-\frac{1}{2-\rho}}{\frac{2}{2-\rho}-1}
\myequationstyle
\left(
\frac{\Gamma(t+1)}{\Gamma\left(t-1+(2-\rho)\right)}\frac{\Gamma\left(k-2+(2-\rho)\right)}{\Gamma(k)} - 1
\right).
\label{eq:ikska}
\end{eqnarray}

\noindent
We, therefore, note that:

\begin{equation}
X_k(t) = m t p_{k}^{(t,G)}
\end{equation}
with:
$$\rho(G) = 2-\frac{1}{G}=\frac{2G-1}{G},$$
or, equivalently:
$$G(\rho)=\frac{1}{2-\rho},$$ and $p_k^{(t,G)}$ given by Eq.~(\ref{eq:process}).

\smallskip
We have thus established an aesthetically pleasing connection between the
3DI model \citep{PNAS2020} and the Gini-stable process \citep{OurGiniLorenz},
and hence the degree of randomness in the impact distribution and the Gini index.
Let us also note that \citet{OurGiniLorenz} have also studied the asymptotic
behaviour of the Lorenz curves generated by this model,
and have shown its relationship to the Pareto Type~II,
exponential, and scaled beta distributions \citep{arnold2015:pareto,pickands1975statistical,marshall2007life}.

\subsection{Lorenz curves}\label{subsec:lorenz}
For a given nonincreasing $N$-vector $\boldsymbol{y}=(y_1,\dots,y_N)$ with
$\sum_{i=1}^N y_i=C$, let $I_{\boldsymbol{y}}(u)$, $0\leqslant u \leqslant 1$, be the continuous curve connecting the origin $(0, 0)$ with the point $(1,C)$ and obtained by joining  the points $\left(i/N, \sum_{j=1}^i y_{N-j+1}\right)$
by line segments, where $i=1,\,2,\,\dots,\,N$.
The function $I_{\boldsymbol{y}}(u)$ may be viewed as an instance of a non-normalised Lorenz curve \citep{egghe2022rank,egghe2023global,egghe2023global2}, also referred to a generalised Lorenz curve  \citep{rousseau2011lorenz}.
In our model, we know the explicit expression for the function $I_{\boldsymbol{p}^{(N,G)}}(u)$ because the cumulative sums of elements in $\boldsymbol{p}^{(N,G)}$ can be written in closed-form.  Generalised Lorenz curves can then be used to define a class of global impact measures, understood as functions that preserve the non-normalised dominance order $\prec$ as defined by \citet{egghe2023global,egghe2023global2} in such a way that $\boldsymbol{p}^{(N,G)}\prec \boldsymbol{p}^{(N',G')}$ holds if and only if $I_{\boldsymbol{p}^{(N,G)}}(u)\leqslant I_{\boldsymbol{p}^{(N',G')}}(u)$ for every $u\in(0,1)$. A function $\gamma$ is called a global impact measure if, for any pair of $N$-tuples $y_i=C p_i^{(N,G)}$ and $y'_i=C'p_i^{(N,G')}$, it satisfies the implication:

\begin{equation*}
    I_{\boldsymbol{y}}(u)\leqslant I_{\boldsymbol{y}'}(u)\;\;\; \Rightarrow\;\;\;\gamma(\boldsymbol{y})\leqslant \gamma(\boldsymbol{y}').
\end{equation*}

\noindent
Note that if $G=G'$ and $C\neq C'$, then the vectors $\boldsymbol{y}$ and $\boldsymbol{y}'$  have the same Lorenz curves, but different non-normalised Lorenz curves. As well known, the Lorenz ordering is invariant to scale transformations  \citep{egghe1991transfer,rousseau1992concentration}. Since in the present study we limit our attention to concentration (and not global impact) measures, it is sufficient to only consider the Lorenz curve in its normalised version. Lorenz curves are the basis of a concentration theory and related acceptable measures of concentration (\citealp[p.~310]{rousseau2018becoming}; \citealp{egghe2001symmetric}).

\section{Measuring Inequality}\label{sec:inequality}

\paragraph{Majorisation order.}

Given two $N$-vectors $\boldsymbol{p}$ and $\boldsymbol{p}'$
with $\sum_{i=1}^N p_i=\sum_{i=1}^N p_i'$,
we define the \emph{majorisation order} $\preceq$
in such a way that $\boldsymbol{p} \preceq \boldsymbol{p}'$ if and only if
for all $k=1,\dots,N$ it holds $F_k \le F_k'$,
where $F_k = \sum_{i=1}^k p_{[i]}$ and $F_k' = \sum_{i=1}^k p_{[i]}'$
are the sums of the $k$ greatest elements in
$\boldsymbol{p}$ and $\boldsymbol{p}'$, respectively \citep{marshall2011majorisation}.
In other words, it is the extension of the standard componentwise relation, $\le$,
applied over the consecutive cumulative sums of ordered items.
Let us note that in our model, such cumulative sums
can be expressed using the following simple formula:

\begin{eqnarray}\label{eq:cdf}
&&F_k^{(N,G)}=\sum_{i=1}^k p_i^{(N,G)}
\myoptionalpart{
    \\
    &=&
    \nonumber
}{
    =
}
\begin{cases}
\myequationstyle
\frac{1-G}{2G-1} \left(
\frac{G}{1-G}\frac{\Gamma(N)}{\Gamma(N-1+1/G)} \frac{\Gamma(k-1+1/G)}{\Gamma(k)}
-\frac{k}{N}
\right),
&
\myoptionalpart{}{\text{if }}
G\neq\frac{1}{2},\\
\displaystyle\frac{k}{N} \left(
1+H_N-H_k
\right),
&
\myoptionalpart{}{\text{if }}
G=\frac{1}{2}.
\\
\end{cases}
\end{eqnarray}

\paragraph{Monotonicity w.r.t.~the majorisation order.}
We can show that for any $N$, $\boldsymbol{p}^{(N,G)}$ as a function of $G$
is monotone with respect to the majorisation order.

\begin{theorem}\label{Thm:majorisation}
For any $G \le G'$ and $N$, it holds $\boldsymbol{p}^{(N,G)} \preceq \boldsymbol{p}^{(N,G')}$.
\end{theorem}

\noindent
The proof is given by \citet{OurGiniLorenz}, where they discuss the Gini-stable process in the context of the more general Lorenz ordering.

\paragraph{Inequality indices.}
Let us now consider any function $\phi$ that maps the ordered normalised
$N$-vectors to the set of real numbers.
We say that $\phi$ is \textit{Schur-convex}, if for any
$\boldsymbol{p} \preceq \boldsymbol{p}'$ it holds
${\phi}(\boldsymbol{p})\le{\phi}(\boldsymbol{p}')$.
This is equivalent to $\phi$'s being increasing in each $F_1,\dots,F_N$
when re-expressed in terms of cumulative sums of ordered elements;
see \citep{BeliakovETAL2016:penaltyinequality}.

Any Schur-convex function normalised such that $\phi(1,0,\dots,0)=1$
and $\phi(1/N,\dots,1/N)=0$ is called an \textit{inequality index};
see \citep{shorrocks1987transfer} and \citep{marshall2011majorisation,lambert2001distribution,chantreuil2011inequality,zheng2007unit,bosmans2016consistent}. The Gini index is an example function fulfilling these properties. Below we recall some other noteworthy inequality indices (e.g., \citealp{ciommi,mehran,imedioolmedo,parasite,%
BeliakovETAL2016:penaltyinequality}).
Then, we derive the formulae for different inequality indices as functions of $N$ and $G$ in our model. This will enable us to compute one index based on any other one.

\paragraph{Bonferroni's index.}
\label{subsec:bonferroni-index}

For any normalised $N$-vector, the Bonferroni index
\citep{Bonferroni1930:index} is given by:
\begin{align}\nonumber
\Bonf(p_1,\dots,p_N) =&\nonumber
\frac{
N \sum_{i=1}^{N}  \left( 1-\sum_{j=1}^i \frac{1}{N-j+1} \right) p_{i}
}{
N-1 %
}
\myoptionalpart{
    \\
    =&
}{
    =
    \frac{
    N \sum_{i=1}^{N-1} \frac{1}{N-i} \sum_{j=1}^i p_j
    }{
    N-1 %
    }
    - N
    + \frac{N}{N-1} \sum_{i=1}^N \left(1-\frac{1}{N-i+1}\right)
    \\
    =&\frac{
    N \sum_{i=1}^{N} \frac{1}{i} \sum_{j=1}^{N-i} p_j
    }{
    N-1 %
    }
    + \frac{N\left(1-\sum_{j=1}^N\frac{1}{j}\right)}{N-1}
    =
}
\frac{
N}{
N-1}\left(
\sum_{k=1}^{N-1} \frac{F_k}{N-k}
+ 1-H_N\right).\nonumber%
\end{align}

\noindent
Substituting  $F_k=F_k^{(N,G)}$ from Eq.~(\ref{eq:cdf}) for $\frac{1}{2}\neq G \in (0,1)$, we get:

\begin{align}\nonumber
\Bonf\left(\boldsymbol{p}^{(N,G)}\right)
\myoptionalpart{=}{
=&\nonumber
\frac{
N}{
(N-1)(2G-1)}\left(
G \sum_{k=1}^{N-1}
\frac{\Gamma(N)\Gamma(k-1+1/G)}{(N-k)\Gamma(N-1+1/G)\Gamma(k)}
-\frac{1-G}{N}\sum_{k=1}^{N-1} \frac{k}{N-k}
+ 1-H_N\right)\\ \nonumber
=&\frac{ %
N}{
N-1}\left(
\frac{G}{2G-1}\left(  H_{N+1/G-2}-H_{1/G-1} \right)
-\frac{1-G}{(2G-1)N}\left(
N H_{N-1}-N+1
\right)
+ 1-H_N
\right)\\\nonumber
=&\frac{
N}{
N-1}\left(
\frac{G}{2G-1}\left(  H_{N+1/G-2}-H_{1/G-1}  \right)
-\frac{1-G}{2G-1}\left(
 H_{N}-1
\right)
+ 1-H_N
\right)\\
=&\frac{
N}{
N-1}\frac{G}{2G-1}\left(
 H_{N+1/G-2}-H_{1/G-1}
-H_{N}+1
\right)\nonumber%
\\
=&
}
\frac{
N}{
N-1}\sum_{k=2}^N\frac{1}{k(k+1/G-2)}.\nonumber%
\end{align}

\noindent
For $G=\frac{1}{2}$, we get:

\begin{equation}
    \Bonf\left(\boldsymbol{p}^{(N,1/2)}\right) =\frac{
N}{
N-1}\sum_{k=2}^N\frac{1}{k^2}.\nonumber%
\end{equation}

\begin{remark}
Unlike in the Gini index's case,
the Bonferroni index depends on the vector length $N$.
However, in the limit as $N\to\infty$, we have:

\begin{equation*}
    \Bonf\left(\boldsymbol{p}^{(N,G)}\right)
    \stackrel{N\to\infty}{\longrightarrow}
    \begin{cases}
    \displaystyle\frac{G}{2G-1}\left(1-H_{1/G-1}\right), & \mathrm{for\; }G\neq\frac{1}{2}, \\
    \quad\\
    \displaystyle\frac{\pi^2-6}{6}, & \mathrm{for\; }G=\frac{1}{2}. \\
    \end{cases}
\end{equation*}
\end{remark}

\paragraph{De Vergottini's index.}
\label{subsec:vergottini-index}

The De~Vergottini index \citep{verg1,verg2} is defined as:

\begin{align}\nonumber
\Verg(p_1,\dots,p_N) = &
\frac{1}{\sum_{i=2}^N \frac{1}{i}} \left(
    \sum_{k=1}^N \frac{F_k}{k}   - 1
\right)
\myoptionalpart{
\\
=&
}{
=
\frac{1}{\sum_{i=2}^N \frac{1}{i}} \left(
   \sum_{i=1}^N  \sum_{j=i}^{N} \frac{p_i}{j}   - 1
\right)
\\
=&
\frac{1}{H_N-1}
    \sum_{j=1}^N p_j \left(H_N-H_{j-1}-1\right)
=
}
1-  \frac{\sum_{k=1}^N p_k H_{k-1}}{H_N-1}
     .\nonumber%
 \end{align}

\noindent
Computing the value of the following sum in our model
(taking $p_k=p_k^{(N,G)}$ from Eq.~(\ref{eq:pk})
with $\frac{1}{2}\neq G \in (0,1)$):

\begin{align*}
\sum_{k=1}^N p_k^{(N,G)} H_{k-1} =&
\myequationstyle
H_N-1 +\frac{G\left(\frac{G-1}{N}+1-2G\right)}{(G-1)(1-2G)}-\frac{\Gamma(1/G-2)\Gamma(N)}{\Gamma(N-1+1/G)},
\end{align*}

\noindent
leads to:

\begin{align*}
    \Verg\left(\boldsymbol{p}^{(N,G)}\right)= %
\myequationstyle
     \frac{1}{H_N-1}\left(\frac{\Gamma(1/G-2)\Gamma(N)}{\Gamma(N-1+1/G)}+\frac{G}{(2G-1)N}-\frac{G}{G-1}\right).
\end{align*}

\noindent
Furthermore, for $G=\frac{1}{2}$, we get:

\begin{equation*}
    \Verg\left(\boldsymbol{p}^{(N,1/2)}\right) = \frac{N-H_N}{N
    \left(H_N-1\right)}.
\end{equation*}

\begin{remark}\label{remark:vinfty}
For $G<1$, the De Vergottini index in our model converges to zero as $N\to\infty$.
However, the convergence rate is extremely slow;
compare Figure~\ref{fig:GvsOthers}.
\end{remark}

\paragraph{Hoover's index.}
\label{subsec:hoover-index}

The Hoover index (\citealp{hoover}; also known as the Robin Hood index)
can be thought of as the normalised Manhattan distance
to the perfectly equal vector:

\begin{equation*}
\Hoov(p_1,\,\dots,\,p_N)=
\frac{N}{2(N-1)} \sum_{k=1}^N\left|p_{k}-\frac{1}{N}\right|.
\end{equation*}

\noindent
We can simplify it by determining:

\begin{equation}
\nu = \max\left\{j: p_j \ge \frac{1}{N} \right\},
\label{eq:nu}
\end{equation}

\noindent
and then writing:

\begin{align*}
\myoptionalpart{&}{}
    \Hoov(p_1,\,\dots,\,p_N)
\myoptionalpart{\\}{}
=&
\myequationstyle
\frac{N}{2(N-1)}\left(
\displaystyle\sum _{k=1}^{\nu}\left(p_k-\frac{1}{N}\right)-\displaystyle\sum _{k=\nu+1}^N\left(p_k-\frac{1}{N}\right)
\right)
\\
=&
\myequationstyle\myoptionalpart{\scriptsize}{}
\frac{N}{2(N-1)}\left(
\displaystyle\sum _{k=1}^{\nu}\left(p_k-\myequationstyle\frac{1}{N}\right)
-
\left(
    \underbrace{\displaystyle\sum _{k=1}^N\left(p_k-\myequationstyle\frac{1}{N}\right)}_{=0}
    -
    \displaystyle\sum _{k=1}^\nu\left(p_k-\myequationstyle\frac{1}{N}\right)
\right)
\right)
\myoptionalpart{\\ =&}{=}
\frac{1}{N-1} \left(N F_{ \nu }-  \nu  \right).
\end{align*}

\noindent
Moreover, for $\frac{1}{2}\neq G\in(0,1)$, the above  results in:
\begin{align*}
\Hoov\left(\boldsymbol{p}^{(N,G)}\right)
=&\myequationstyle\frac{G}{2G-1}\left(
\frac{N}{N-1}\frac{\Gamma(N)}{\Gamma(N-1+1/G)} \frac{\Gamma( \nu -1+1/G)}{\Gamma( \nu )}
-\frac{ \nu }{N-1}\right).\label{eq:H_solution}
\end{align*}

Please note that  $\nu$ is uniquely determined by the values of $G$ and $N$. Thus, the above does not introduce new parameters to the model.

\begin{remark}\label{rem:PP}
For large $N$ and $G\neq \frac{1}{2}$,
we can express the solution of the continuous approximation
to Eq.~(\ref{eq:nu}), i.e.,
$p_\nu^{(N,G)}=1/N$, as:
\begin{equation*}
  \nu\approx N \left(\frac{1-G}{G}\right)^{G/(2G-1)};
\end{equation*}

\noindent see \citet{pricepareto} for a proof using a slightly different notation
in the case of $G>\frac{1}{2}$, which can be easily extended to  $G<\frac{1}{2}$. Thanks to this result, we obtain a compact asymptotic formula for the Hoover index:

\begin{align*}
    \Hoov\left(\boldsymbol{p}^{(N,G)}\right)
    \stackrel{N\to\infty}{\longrightarrow}&
    \myequationstyle
    \frac{G}{2G-1}\left(
\left(\frac{1-G}{G}\right)^{(1-G)/(2G-1)}
-\left(\frac{1-G}{G}\right)^{G/(2G-1)}\right)
\myoptionalpart{\\ =&}{=}
\left(\frac{1-G}{G}\right)^{(1-G)/(2G-1)}.
\end{align*}

\noindent
Also note that for $G=\frac{1}{2}$, we can obtain $\nu$ using
the Euler--Maclauren formula, which yields:

\begin{equation*}
    p_k^{(N,1/2)}\approx \frac{1}{N}\left(\frac{1}{2k}+\frac{1}{2N}-\log(k/N) \right).
\end{equation*}

\noindent
This allows us to solve equation $p^{(\nu,1/2)}=1/N$ through
the principal branch of the Lambert $\mathcal{W}$ function,
denoted $\mathcal{W}_0$, leading to:

\begin{align*}
    \nu
    \approx &
    \frac{1}{2}\left( \mathcal{W}_0\left(\myequationstyle\frac{\exp\left(1-\frac{1}{2 N}\right)}{2 N}\right)\right)^{-1}
    =
    \myequationstyle
    \frac{1+e}{2 e}+ \frac{N}{e}+\frac{1-e^2}{8 e N}+O\left(N^{-2}\right)
    \myoptionalpart{\\ =&}{=}
    \frac{(2N+1)e^{-1}+1}{2 }+O\left(N^{-1}\right),
\end{align*}

\noindent
leading us to:

\begin{align*}
    \Hoov\left(\boldsymbol{p}^{(N,1/2)}\right)
    \stackrel{N\gg 1}{\approx} &
    \myequationstyle
    \frac{N}{N-1}  F^{(N,1/2)}_{
    \left\lfloor \frac{2Ne^{-1}+e^{-1}+1}{2}  \right\rfloor
    }
    -
    \frac{1}{N-1}
    \left\lfloor \frac{2Ne^{-1}+e^{-1}+1}{2}  \right\rfloor
    \myoptionalpart{\\}{}
    \stackrel{N\to\infty}{\longrightarrow}
    \myoptionalpart{&}{}
    e^{-1}.
\end{align*}

\end{remark}

\paragraph{The $P_q$ indices.}
\label{subsec:pq-index}

For a fixed $q\in(0,1)$, the $P_q$ index is defined as:

\begin{equation}\label{Eq:pq}
\Pequ_q(p_1,\,\dots,\,p_N)=\frac{\sum_{j=1}^{\lfloor q
N\rfloor}p_j-q}{1-q}=\frac{F_{\lfloor qN \rfloor}-q}{1-q}.
\end{equation}

\noindent
This index is a normalised version of the percentage of accumulated
mass in $q 100\%$ of the top elements (compare the famous Pareto 80/20 rule).
Note that $q$ is as parameter of the index, and not a new parameter of the model;
in this sense, $P_q$ is a whole family of indices.

In our model,  for $\frac{1}{2} \neq G \in(0,1) $, the $P_q$ index is equal to:

\begin{align*}
\Pequ_q\left(\boldsymbol{p}^{(N,G)}\right) %
=&
\myequationstyle
\frac{G}{(2G-1)(1-q)}\left(\frac{\Gamma(N)}{\Gamma(N-1+1/G)} \frac{\Gamma(\lfloor qN \rfloor-1+1/G)}{\Gamma(\lfloor qN \rfloor)}
-\frac{(1-G)\lfloor qN\rfloor}{GN}\right)-\frac{q}{1-q}\approx\\
\approx&
\myequationstyle
\frac{G}{(2G-1)(1-q)}\left(\frac{\Gamma(N)}{\Gamma(N-1+1/G)} \frac{\Gamma(\lfloor qN \rfloor-1+1/G)}{\Gamma(\lfloor qN \rfloor)}
-q\right).
\end{align*}

\noindent
Furthermore, if $G=\frac{1}{2}$, then it holds:

\begin{equation}
    \Pequ_q\left(\boldsymbol{p}^{(N,1/2)}\right)
    = \frac{q}{1-q}(H_N-H_{\lfloor qN \rfloor}).
\end{equation}

\noindent

\begin{remark}
In the limit as $N\to\infty$, we have:

\begin{equation*}
    \Pequ_q\left(\boldsymbol{p}^{(N,G)}\right)
    \stackrel{N\to\infty}{\longrightarrow}
    \begin{cases}
    \displaystyle\frac{G}{(2G-1)(1-q)} \left(
q^{1/G-1}
-q
\right), & \myoptionalpart{}{\text{for }}G\neq\frac{1}{2}, \\
    \quad\\
    \displaystyle\frac{-q\log(q)}{1-q}, & \myoptionalpart{}{\text{for }}G=\frac{1}{2}. \\
    \end{cases}
\end{equation*}
\end{remark}

\paragraph{Indices as functions of one another.}\label{subsec:equivalence}

To sum up, we have shown that in our model,
the formulae for the Bonferroni, De Vergottini, Hoover, and $P_q$
indices can be expressed as the following functions of $N$ and $G$:

\begin{align}
    \label{eq:final-equations1}
\tilde{B}(N,G) =& %
\myequationstyle
\frac{
N}{
N-1}\sum_{k=2}^N\frac{1}{k(k+1/G-2)},
\\
    \label{eq:final-equations2}
\tilde{V}(N, G)=& %
\myequationstyle
    \frac{1}{H_N-1}\left(\frac{\Gamma(1/G-2)\Gamma(N)}{\Gamma(N-1+1/G)}+\frac{G}{(2G-1)N}-\frac{G}{G-1}\right),
\\
\tilde{H}(N, G) =& %
\myequationstyle
\frac{N}{N-1} \left(\frac{G}{2G-1}
\frac{\Gamma(N) }{\Gamma(N-1+1/G) }
\frac{\Gamma(\nu-1+1/G)}{\Gamma(\nu)}
-  \frac{\nu}{N}  \right),
\qquad(\text{with }\nu = \max\{j: p_j \ge {1}/{N} \})
    \label{eq:final-equations3}
\\
\tilde{P}_q(N, G) =& %
\myequationstyle
\frac{G}{(2G-1)(1-q)}\left(\frac{\Gamma(N)}{\Gamma(N-1+1/G)} \frac{\Gamma(qN-1+1/G)}{\Gamma(qN)}
-q\right),\label{eq:final-equations4}
\end{align}

\noindent
respectively (for readability, we only included the case of $G\neq 1/2$).
Figure~\ref{fig:GvsOthers} depicts them for different $N$s.

\begin{figure}[bht!]
    \begin{center}
    \includegraphics[width=0.4\textwidth]{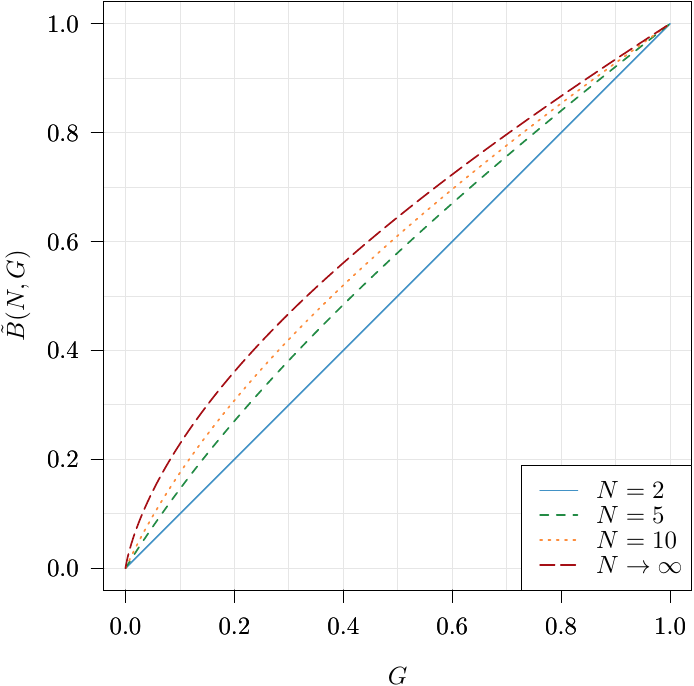}
    \includegraphics[width=0.4\textwidth]{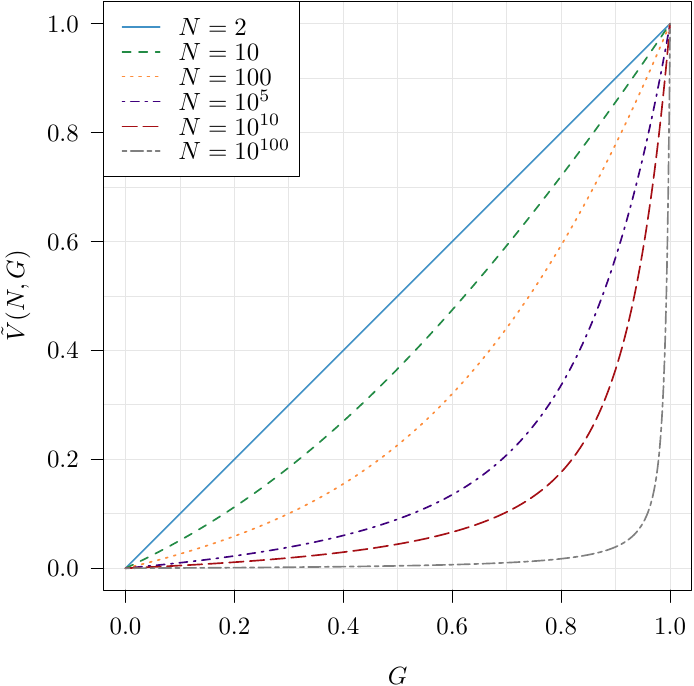}\\
    \vspace{1em}
    \includegraphics[width=0.4\textwidth]{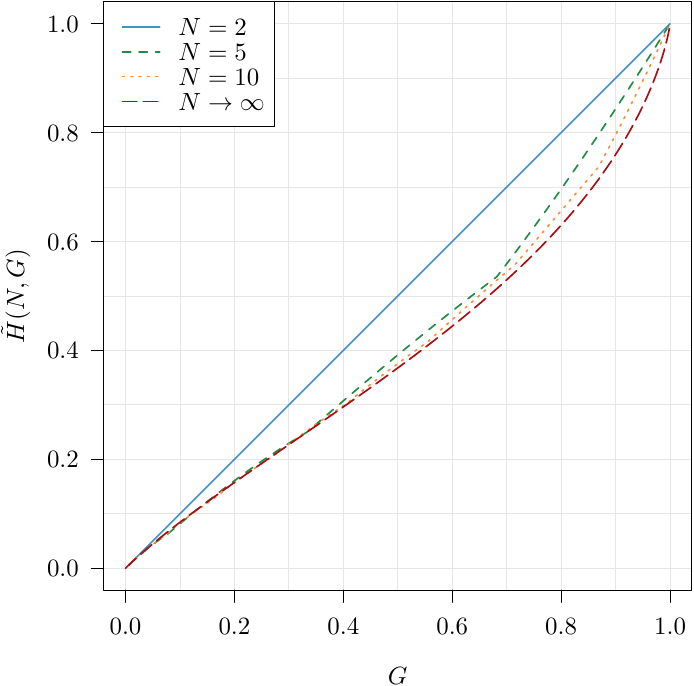}
    \includegraphics[width=0.4\textwidth]{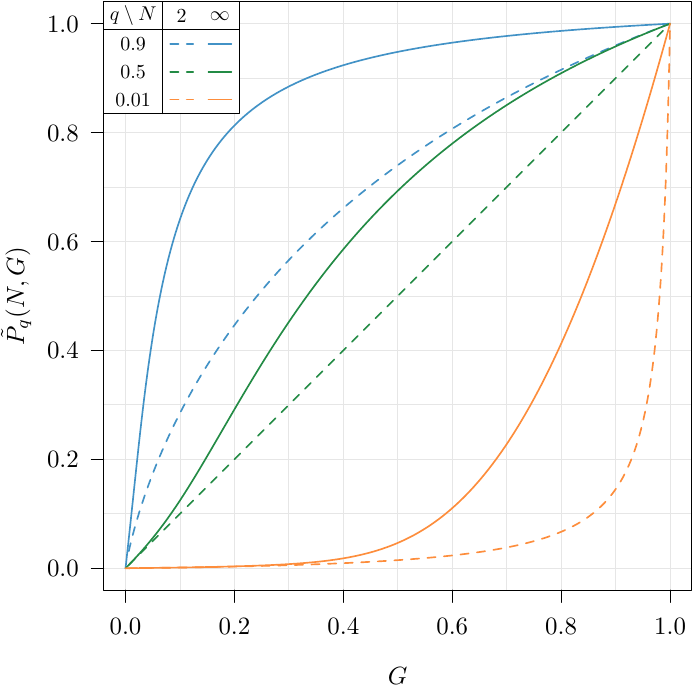}
    \end{center}
    \caption{%
        Functional dependence between the Gini index and other inequality measures
        for different sample sizes $N$ in our model.
        The indexes are one-to-one functions of one another.
        Therefore, similar plots could be drawn for $V$ as a function of $B$,
        $H$ as a function of $P_q$, etc.
    }%
    \label{fig:GvsOthers}
\end{figure}

We see that $G$ has a higher sensitivity to inequality for low levels of the inequality
scale as compared to the $H$, $V$ and $P_q$ for $q$ close to $0$,
and a smaller sensitivity at the opposite end of the inequality scale. This confirms and extends the assessments of \citet{ciommi}.
Also, $B$ and $P_q$ for $q$ close to $1$ evidence a higher
sensitivity to inequality as compared to $G$ for low levels of the
inequality scale, and quite the opposite for high levels of inequality.

We note that the foregoing are all strictly increasing, continuous functions of $G$.
Thus, based on the derived formulae, all the indices can be expressed
as one-to-one functions of one another.
For instance, given some value of the Bonferroni index $B$,
we can obtain the underlying $G^* = \tilde{B}^{-1}(N, B)$
and then compute, say, $\tilde{V}(N, G^*)$.
Even though the analytic formulae for the inverses do not exist,
they can easily be solved numerically.
In this sense, we can say that -- \textit{in our model} --
all the aforementioned inequality indices are equivalent.
Similar derivations can be performed for many other inequality indices,
although they might not necessarily enjoy analytic solutions.

\section{Experiments}\label{sec:experiments}

Looking beyond our simple iterative process,
we know that uncountably many normalised vectors yield
a specific Gini, Bonferroni, or any other index.
After all, these measures were introduced to respond to the different needs
of the practitioners \citep{imedioolmedo,parasite}.
Some of them are, for example, more responsive to
the increasing
of the amount of mass in the tail of the distribution
(via the principle of progressive transfers) than others.
In particular, as reported by \citet{ciommi}, the Gini, Bonferroni,
and De Vergottini indices belong to the class of linear measures introduced by
\citet{mehran}. They note that \textit{for the Bonferroni and De Vergottini
indices, the effect of a transfer also depends on the position of individuals,
making the Bonferroni index more sensitive to transfers that occur at the lower
end of the income distribution and the De Vergottini index more sensitive to
variations among the richest}.

Theoretical models are merely approximations of the real-world phenomena under
scrutiny, but some models are more useful than others.
\citet{OurGiniLorenz} have already noted that the Gini-stable model
provides a good fit to a variety of informetric and other data,
even though the assumption about the Gini index's being constant throughout
all iterations is merely an idealisation, which does not necessarily
have to hold throughout the whole agent's lifespan.

Therefore, we should be interested in verifying how well our formulae
approximate the hidden relationships present in real datasets.
Let us thus consider citation data from the RePEc database
(Research Papers in Economics; see \url{https://citec.repec.org/}), which
features $66{,}347$ authors and $1{,}843{,}967$ papers.
In the data cleansing step, we have omitted the authors who published less than
$5$ cited papers and whose $h$-index was less than $3$ for such samples
are too small to make the analysis robust enough.
This resulted in $n=36{,}425$ citation records of the form
$\left(x_1^{(i)},\dots,x^{(i)}_{N^{(i)}}\right)$, where
$N^{(i)}$ gives the total number of items published by the $i$-th author
and $x_k^{(i)}$ gives the number of citations to their $k$-th most cited works.

Figure~\ref{fig:repec_3vectors_loglog} shows
three example (quite representative of the whole database)
vectors from the RePEc database:
observed (points) and predicted (lines) values for $p_k$ (left)
and their non-normalised versions, $x_k$ (right).
We see a good fit over most parts of the data domain. Testing all vectors
with the discrete Kolmogorov--Smirnov test~\citep{taylor2011R:dgof} resulted in rejecting the null hypothesis that data follow our distribution only in $1\%$ of the cases
(at the significance level of $0.05$).
This comes as no surprise, as the proposed model is equivalent
to the 3DI model~\citep{PNAS2020} and we have already
seen its usefulness in the case of modelling citations to computer science papers. Nevertheless, we should note that the statistical power of the K-S test is not particularly high.

\afterpage{\clearpage}

\begin{figure}[t!]
    \begin{center}
        \includegraphics[width=0.45\textwidth]{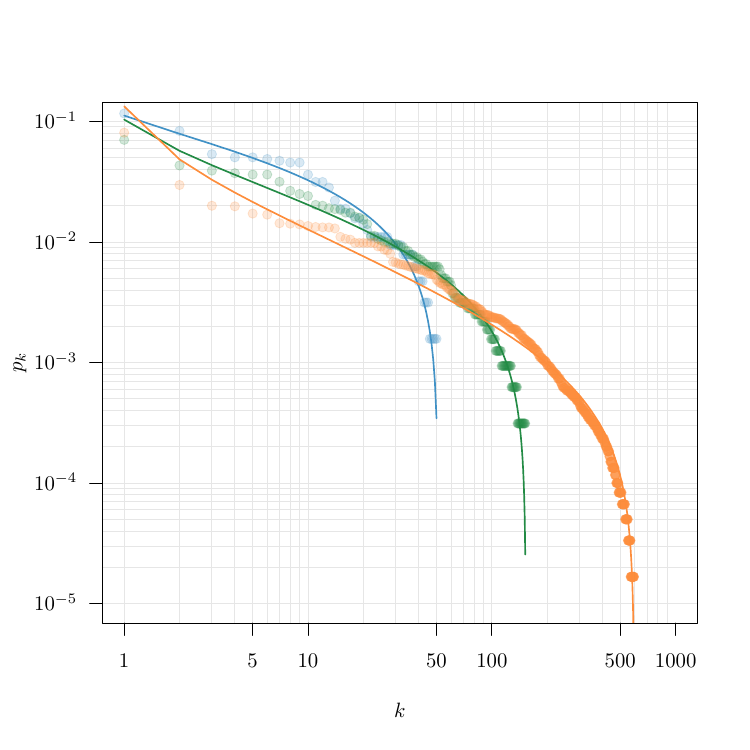}
        \includegraphics[width=0.45\textwidth]{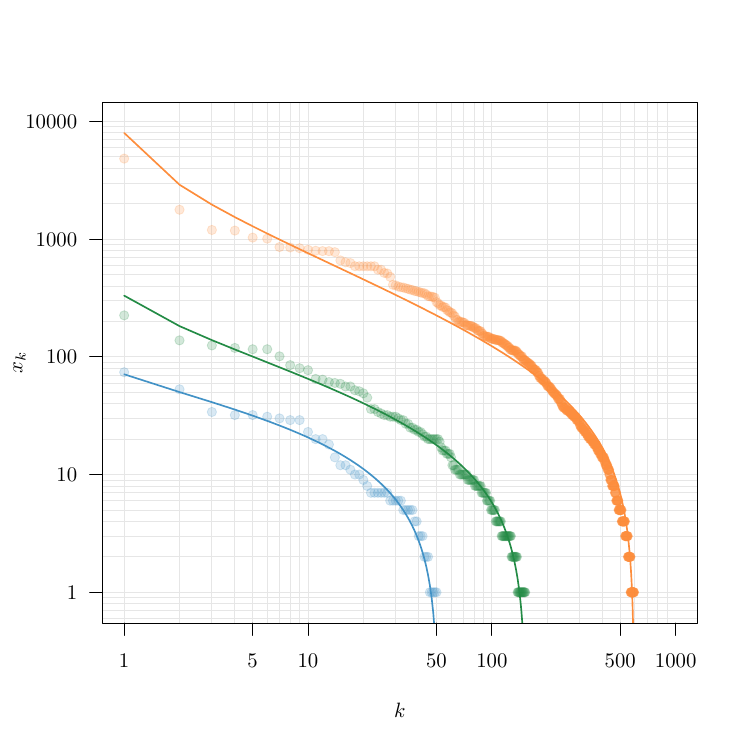}
    \end{center}
    \caption{%
    Three example citation vectors (left: normalised, right: original) and
    the corresponding fitted models (Eq.~(\ref{eq:process}));
    note the log-log scale. We note a very good fit in each case.
    }
    \label{fig:repec_3vectors_loglog}
\end{figure}

For each author, we computed their actual
(observed) Gini index
$\hat{G}^{(i)} = \Gini(p_1^{(i)},\dots,p^{(i)}_{N^{(i)}})$
with $p_k^{(i)} = X_k^{(i)}/\sum_{j=1}^{N^{(i)}} X_j^{(i)}$.
Knowing the value of $\hat{G}^{(i)}$, based on the derived
formulae (see Eqs.~\eqref{eq:final-equations1}--\eqref{eq:final-equations4}),
we computed the predicted values
of the indices:
$\tilde{B}(N, \hat{G}^{(i)})$, $\tilde{V}(N, \hat{G}^{(i)})$, $\tilde{V}(N,
\hat{G}^{(i)})$ and $\tilde{P}_{q}(N, \hat{G}^{(i)})$.
This way, we can compare these approximated values
with the observed ones, i.e.,
$\hat{B}^{(i)} = \Bonf(p_1^{(i)},\dots,p^{(i)}_{N^{(i)}})$,
$\hat{V}^{(i)} = \Verg(p_1^{(i)},\dots,p^{(i)}_{N^{(i)}})$,
$\hat{H}^{(i)} = \Hoov(p_1^{(i)},\dots,p^{(i)}_{N^{(i)}})$, and
$\hat{P}_q^{(i)} = \Pequ(p_1^{(i)},\dots,p^{(i)}_{N^{(i)}})$.

Recall that Figure~\ref{fig:GvsOthers} describes the
theoretical relationships between the Gini index and other metrics.
If, overall,  our model describes the real data well,
we should expect to see these dependencies in the case of the RePEc vectors too.
Figure~\ref{fig:repec_indG_largeN} presents a scatter plot
of the values of different indices
as functions of $\hat{G}$ for all vectors with $N>200$
(the De Vergottini index was not included as it approaches $0$ for large $N$s;
see Remark \ref{remark:vinfty}).
Ideally, they should lie close to the theoretical curves
(depicted as well). And this is approximately the case.

\begin{figure}[p!]
    \begin{center}
        \includegraphics[width=\myoptionalpart{0.95}{0.45}\linewidth]{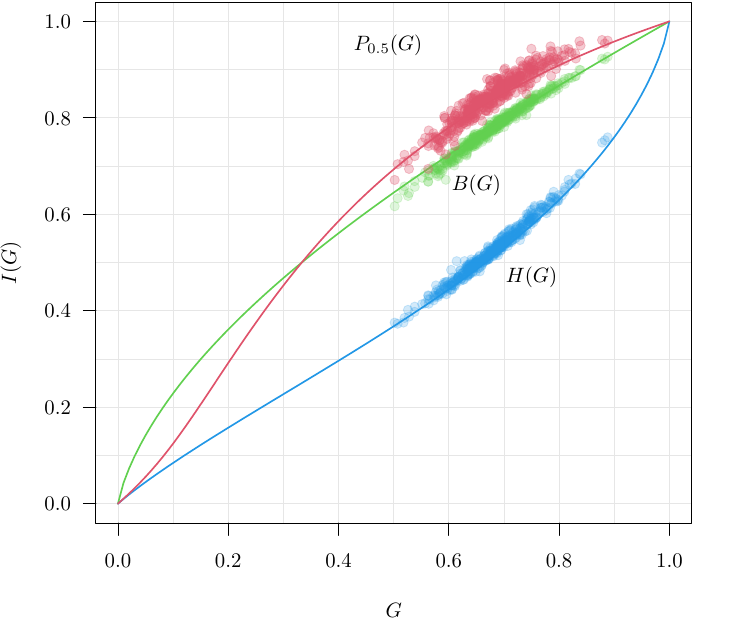}
    \end{center}
    \caption{%
        Inequality indices as functions of the Gini index
        for all RePEc citation vectors with $N>200$ (points).
        They match the derived theoretical curves (lines) quite well,
        confirming the usefulness of our model.
    }%
    \label{fig:repec_indG_largeN}
\end{figure}

\begin{figure}[p!]
    \begin{center}
        \includegraphics[width=0.4\textwidth]{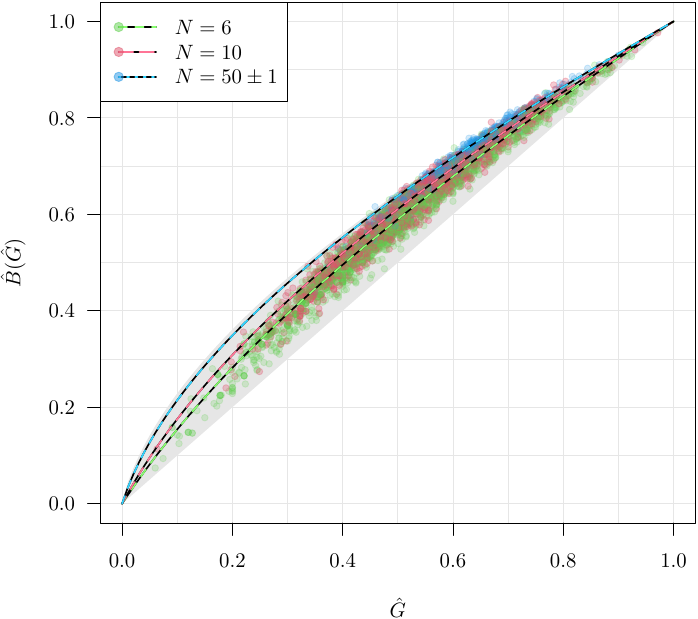}
        \includegraphics[width=0.4\textwidth]{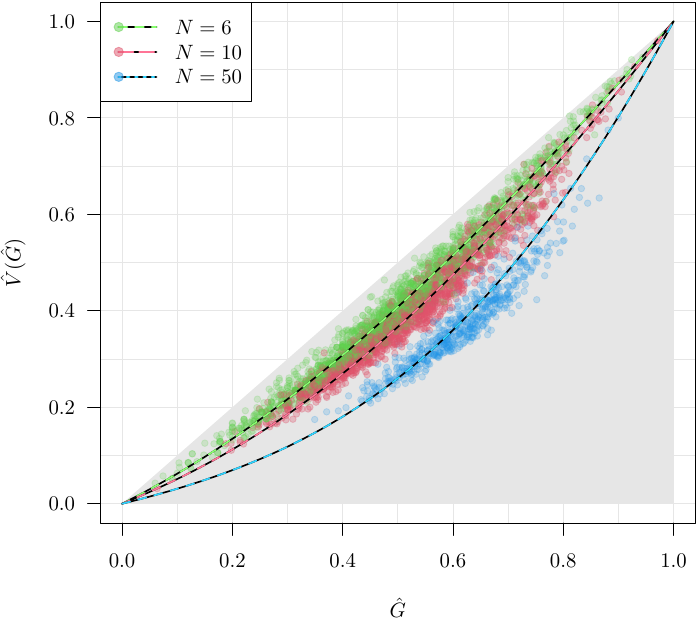}

        \includegraphics[width=0.4\textwidth]{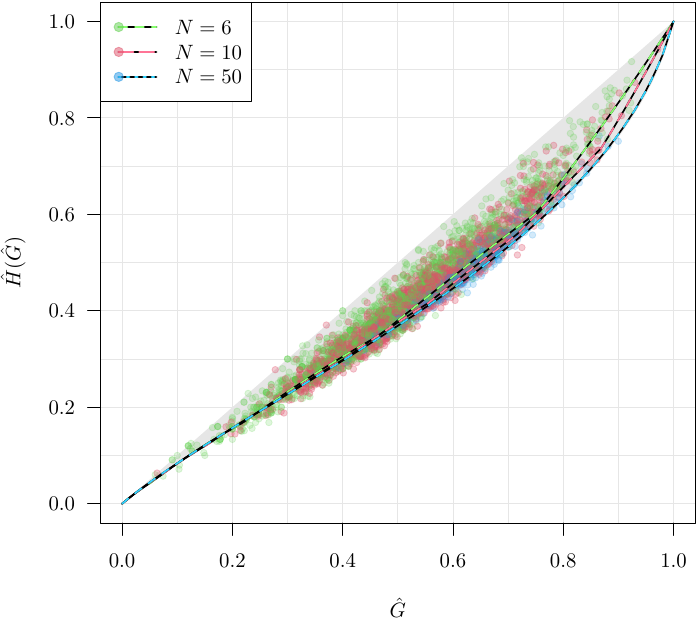}
        \includegraphics[width=0.4\textwidth]{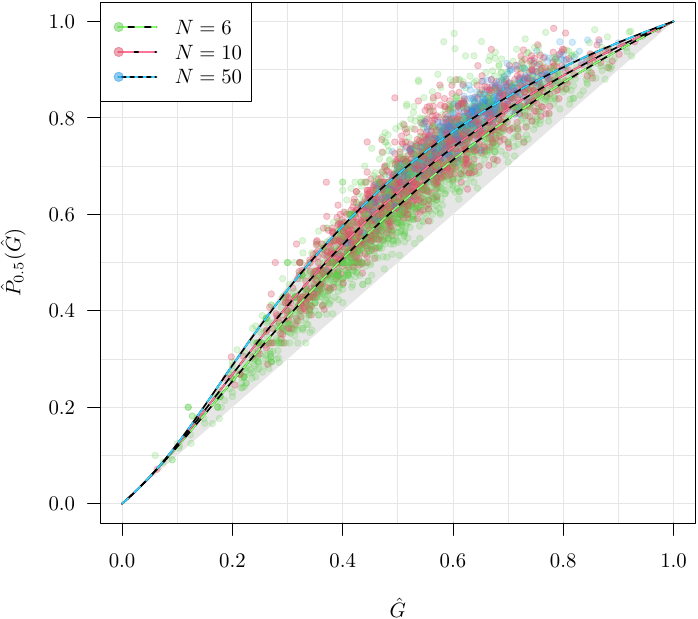}
    \end{center}
    \caption{%
        Predicted ($\tilde{B}(N, \hat{G})$ vs $\hat{G}$,
        $\tilde{V}(N, \hat{G})$ vs $\hat{G}$, etc.; thick curves)
        and
        observed ($\hat{B}$ vs $\hat{G}$ etc.; points)
        values of different inequality measures
        for the RePEc citation records of different lengths $N$.
        Grey areas represent the values which can be obtained from our model.
    }%
    \label{fig:repec}
\end{figure}

Furthermore, Figure~\ref{fig:repec} presents similar results, but for vectors of
lengths $N=6$ (green), $N=10$ (pink), and $N=50\pm 1$ (blue;
the $\pm 1$ part is to increase the number of data points).
Additionally, we coloured the areas representing all of the possible values
which could be obtained for vectors generated from our model. In other words,
for any vector, we expect the pairs $(\hat{G}, \hat{B})$,
$(\hat{G}, \hat{V})$, etc.~to lie in the grey zone.
This is true for the vast majority of the real data points.

\FloatBarrier

\begin{figure}[t!]
    \centering
    \includegraphics[width=0.49\linewidth]{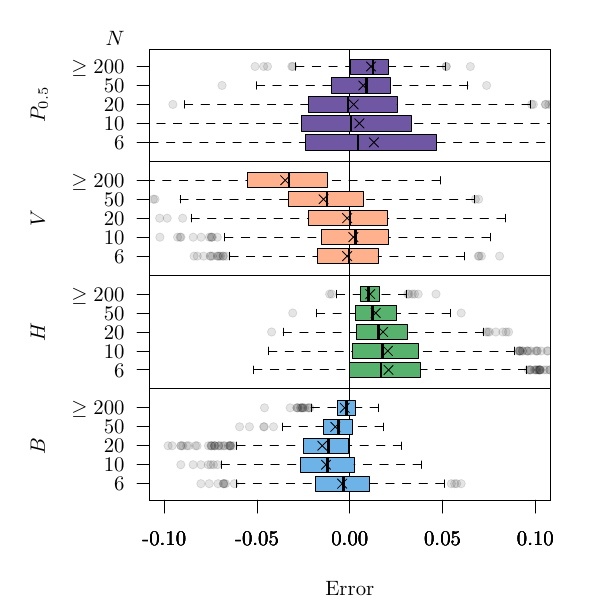}
    \includegraphics[width=0.49\linewidth]{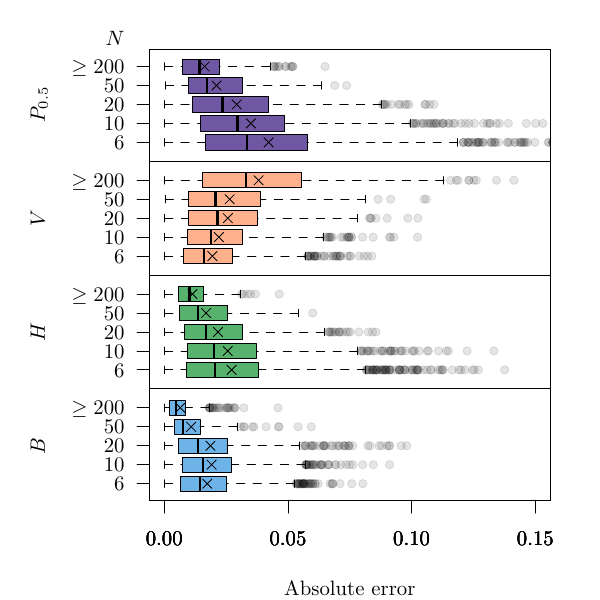}
    \caption{Box plots of prediction errors, i.e., $\hat{I} - \tilde{I}$
    (left) and absolute prediction errors, i.e., $|\hat{I} - \tilde{I}|$
    (right) for different values of $N$ and different indices; RePEc data.}
    \label{fig:repec-errors-boxplot}
\end{figure}

    \begin{table}[t!]
    \centering
    \caption{Spearman correlation coefficients between index
        $\hat{I}\in\{\hat{H}, \hat{B}, \hat{V}, \hat{P}_{0.5}\}$ and
        index $\hat{G}$ calculated directly from the RePEc data for subsets of
        vectors of a given length and for all vectors (last row).}
    \label{tab:repec-spearman-data}
    \begin{tabular}[t]{r|r|r|r|r}
        $N$\textbackslash$\hat{I}$&
        $\hat{H}$&$\hat{B}$&$\hat{V}$&$\hat{P}_{0.5}$ \\\hline
        $5$&\cellcolor[HTML]{FCAB5F}{\textcolor{black}{0.98}}&
        \cellcolor[HTML]{FCA75D}{\textcolor{black}{0.99}}&
        \cellcolor[HTML]{FCA75D}{\textcolor{black}{0.99}}&
        \cellcolor[HTML]{FDAF62}{\textcolor{black}{0.97}}\\\hline
        $10$&\cellcolor[HTML]{FCAB5F}{\textcolor{black}{0.98}}&
        \cellcolor[HTML]{FCAB5F}{\textcolor{black}{0.98}}&
        \cellcolor[HTML]{FCAB5F}{\textcolor{black}{0.98}}&
        \cellcolor[HTML]{FDB869}{\textcolor{black}{0.94}}\\\hline
        $15$&\cellcolor[HTML]{FCAB5F}{\textcolor{black}{0.98}}&
        \cellcolor[HTML]{FCAB5F}{\textcolor{black}{0.98}}&
        \cellcolor[HTML]{FDAF62}{\textcolor{black}{0.97}}&
        \cellcolor[HTML]{FDB567}{\textcolor{black}{0.95}}\\\hline
        $20$&\cellcolor[HTML]{FCAB5F}{\textcolor{black}{0.98}}&
        \cellcolor[HTML]{FCAB5F}{\textcolor{black}{0.98}}&
        \cellcolor[HTML]{FDAF62}{\textcolor{black}{0.97}}&
        \cellcolor[HTML]{FDB567}{\textcolor{black}{0.95}}\\\hline
        $50$&\cellcolor[HTML]{FCAB5F}{\textcolor{black}{0.98}}&
        \cellcolor[HTML]{FCA75D}{\textcolor{black}{0.99}}&
        \cellcolor[HTML]{FDB264}{\textcolor{black}{0.96}}&
        \cellcolor[HTML]{FDBB6C}{\textcolor{black}{0.93}}\\\hline
        $100$&\cellcolor[HTML]{FCA75D}{\textcolor{black}{0.99}}&
        \cellcolor[HTML]{FCA75D}{\textcolor{black}{0.99}}&
        \cellcolor[HTML]{FDB264}{\textcolor{black}{0.96}}&
        \cellcolor[HTML]{FDAF62}{\textcolor{black}{0.97}}\\\hline
        all&\cellcolor[HTML]{FDAF62}{\textcolor{black}{0.97}}&
        \cellcolor[HTML]{FCAB5F}{\textcolor{black}{0.98}}&
        \cellcolor[HTML]{FDDE89}{\textcolor{black}{0.82}}&
        \cellcolor[HTML]{FDBF6F}{\textcolor{black}{0.92}}\\\hline
    \end{tabular}

\end{table}

The prediction errors for different vector lengths are
summarised in Figure~\ref{fig:repec-errors-boxplot}.
The absolute values of errors, i.e., $|\hat{I} - \tilde{I}|$ are usually less
than $0.02$--$0.03$. The errors themselves (bias) are small as well.

We can also be interested in the way they order
the vectors of interest.
To test if all indices give similar rankings, we can calculate the
Spearman correlation coefficient between them.
Table~\ref{tab:repec-spearman-data} presents the
correlations between $\hat{G}$ and $\hat{I}\in\{\hat{H}, \hat{B}, \hat{V},
\hat{P}_{0.5}\}$ for subsets of vectors of a given length.
For vectors of equal lengths,
the obtained rankings are very similar. If we take a look at the whole
data set (the last row), the correlations are still high and the biggest
difference can be seen for index $V$. It
is however not surprising as its dependence on $G$  varies greatly
for different values of $N$.

Overall, we obtained a good fit to our theoretical derivations.
For data approximately following our model, we may approximately assume
that when we consider vectors of similar lengths,
the choice of the inequality measure is secondary, as the indices can be
considered functions of one another. Therefore, in such a case,
we can be faithful to the simplest indicator: the Gini index.

\clearpage

\section{Conclusions}\label{sec:conclusion}

It may be interesting to note that we can easily generalise our Gini-stable model $i\rightarrow p_i^{(N,G)}$, $i=1,\,2,\,\dots,N$,  to take the production factor into account by writing $\boldsymbol{y}=C \boldsymbol{p}^{(N,G)}$, $C>0$.  While $p^{(N,G)}_i$ may be interpreted as the relative number of items in the $i$-th  cell of the IPP, $y_i=C p_i^{(N,G)}$ represents its corresponding absolute value. For example, the vector $\boldsymbol{y}$  may express the number of citations
earned by $N$ articles. In that case, $C$ is the total number of citations. Hence, both elements of the theory on impact by \citet{egghe2022rank,egghe2023global} are independently present in the array model $\boldsymbol{y}=C \boldsymbol{p}^{(N,G)}$.

The discussed process yields $\boldsymbol{p}^{(N,G)}$ that are totally ordered by the majorisation relation $\preceq$. Following \citet{shorrocks1987transfer} (but see also \citealp{marshall2011majorisation,lambert2001distribution,chantreuil2011inequality,zheng2007unit,bosmans2016consistent}), a function $I$ is an index of inequality if and only if it is symmetric and strictly Schur convex.

Hence, for all ordered normalised $N$-vectors
$\boldsymbol{p}\neq\boldsymbol{q}$, if $\boldsymbol{p}\preceq \boldsymbol{q}$,
then $I(\boldsymbol{p})<I(\boldsymbol{q})$, where the direction of the
inequality is uniquely determined by the Pigou–Dalton condition (e.g.,
\citealp{Patty2019}).
In particular, for every vector $\boldsymbol{p}^{(N,G)}$, the parameter $G$ can be interpreted as the
(normalised) Gini index $\Gini(\boldsymbol{p}^{(N,G)})=G$. This implies that $\Gini$ is an order preserving function \citep{marshall2011majorisation}  on $\{\boldsymbol{p}^{(N,G)}\}$ with respect to $G$.

The main requirement to impose to a measure of concentration (inequality) to be ``acceptable'' is that to be coherent to the Lorenz ordering. By construction, the Gini index is acceptable because it is coherent with the Lorenz ordering \citep{egghe2010hirsch}. But other indices are equally acceptable as measures of concentration. Examples of such indices are $B$, $H$, $V$ and $P_q$. How can we compare them with each other? At least within the case of distributions belonging to the Gini-stable family, we showed that it is possible to derive their explicit expressions as one-to-one functions of $G$ and $N$. This allows us to distinguish and compare all these measures of concentration according to their degree of sensitivity to inequality.

For data vectors of similar sizes that follow closely our model, the indices can be considered equivalent.
An analysis of an empirical dataset (citation vectors in economics)
confirmed our results. Of course, we need to keep in mind that our model
is not universal: \citet{OurGiniLorenz} presented the instances
of both fair and poor fits.

The paper focused on an analysis of the Gini-stable process and relations
between inequality measures in it. Even though the empirical data fit the vectors
generated by this model reasonably well, there are some questions that have not been answered yet. In particular, the analysed process assumes that the vectors have the same Gini index over their whole lifespan. The extent to which this idealisation is correct, at least approximately, is yet to be determined, although data that could be used to validate this hypothesis is difficult to obtain. We also assume that wealth is being distributed equally over time, which is also an
assumption that should be tested on real data.

\subsection*{Acknowledgements}
We are indebted to Jose Manuel Barrueco for providing us with a large snapshot
of RePEc (Research Papers in Economics) data.
All data are freely available at \url{http://citec.repec.org/api.html}.

This research was supported by the Australian Research Council Discovery
Project ARC DP210100227 (MG).

We thank the anonymous reviewers whose remarks led to the improvement
of the manuscript.

\subsection*{Conflict of interest}
The authors certify that they have no affiliations with or involvement in any
organisation or entity with any financial interest or non-financial interest
in the subject matter or materials discussed in this manuscript.

\printcredits  %

\bigskip
\noindent
Please cite this paper as:
Bertoli-Barsotti, L., Gagolewski, M., Siudem, G., Żogała-Siudem, B., Equivalence of inequality indices in the three-dimensional model of informetric impact, \textit{Journal of Informetrics} \textbf{18}(4), 101566, 2024, \href{https://doi.org/10.1016/j.joi.2024.101566}{DOI:10.1016/j.joi.2024.101566}.

\end{document}